\documentclass[11pt]{llncs}
\usepackage[a4paper]{geometry}
\usepackage{hyperref}
\usepackage{array}
\usepackage{amssymb}
\usepackage{amsfonts}
\usepackage{amsmath}
\usepackage{mathrsfs}
\usepackage{latexsym}
\usepackage{graphicx}
\usepackage{subfigure}
\usepackage{psfrag}
\usepackage{enumerate}
\usepackage{url}
\usepackage{color}
\usepackage{nicefrac}
\usepackage[utf8]{inputenc}
\usepackage{verbatim}
\usepackage{authblk}
\usepackage{xspace}
\usepackage{textcomp}
 \newtheorem{observation}{Observation}

\newcommand{\ie}{\textit{i.e.}}
\renewcommand{\P}{\mathcal{P}}
\newcommand{\G}{\mathcal{G}}

\newcommand{\F}{\mathcal{F}}

\newcommand{\hg}{\ulcorner}
\newcommand{\hd}{\urcorner}
\newcommand{\bg}{\llcorner}
\newcommand{\bd}{\lrcorner}

\newcommand{\N}{\mathbb{N}}

\title{On independent set on B1-EPG graphs\thanks{This work was supported by the grant EGOS ANR-12-JS02-002-01. This is an extended version of the WAOA 2015 corresponding article.}}
\author{M. Bougeret, S. Bessy, D. Gon\c{c}alves, and C. Paul}
\institute{LIRMM, CNRS \& Université de Montpellier}

\begin{document}

\maketitle

\begin{abstract}
In this paper we consider the Maximum Independent Set problem (MIS) on
$B_1$-EPG graphs. EPG (for Edge intersection graphs of Paths on a
Grid) was introduced in ~\cite{edgeintersinglebend} as the class of
graphs whose vertices can be represented as simple paths on a
rectangular grid so that two vertices are adjacent if and only if the
corresponding paths share at least one edge of the underlying
grid. The restricted class $B_k$-EPG denotes EPG-graphs where every
path has at most $k$ bends.  The study of MIS on $B_1$-EPG graphs has
been initiated in~\cite{wadsMIS} where authors prove that MIS is
NP-complete on $B_1$-EPG graphs, and provide a polynomial
$4$-approximation.  In this article we study the approximability and
the fixed parameter tractability of MIS on $B_1$-EPG.  We show that
there is no PTAS for MIS on $B_1$-EPG unless P$=$NP, even if there is
only one shape of path, and even if each path has its vertical part or
its horizontal part of length at most $3$. This is optimal, as we show
that if all paths have their horizontal part bounded by a constant,
then MIS admits a PTAS.  Finally, we show that MIS
is FPT in the standard parameterization on $B_1$-EPG restricted to
only three shapes of path, and $W_1$-hard on $B_2$-EPG. The status for
general $B_1$-EPG (with the four shapes) is left open.
\end{abstract}


\section{Introduction and related work}
In this paper we consider the Maximum Independent Set (MIS) on $B_1$-EPG
graphs. EPG (for Edge intersection graphs of Paths on a Grid)
was introduced in ~\cite{edgeintersinglebend} as the class of graphs whose vertices
can be represented as simple paths on a rectangular grid so that two vertices are adjacent
if and only if the corresponding paths share at least one edge of the
grid.
More precisely, for every EPG-graph $G=(V,E)$, there exists an
EPG-representation $\langle \P,\G \rangle$ where $\P = \{P_v,
v \in V\}$ is a set of paths on the grid $\G$, and two paths $P_u$,
$P_v$ share a grid edge of $\G$ if and only if $\{u,v\} \in E$ (see example
depicted Figure~\ref{fig:intro_EPG}). Notice that two
paths can cross on a grid vertex without creating an edge.
Thus, given an EPG-representation, the objective of the MIS is to find
the largest set of paths that do not share a common grid edge.
For any integer $k \ge 0$, the class $B_k$-EPG denotes EPG-graphs
having an EPG-representation where every path has at most $k$ bends.
Moreover, the class $X$-EPG $\subseteq$ $B_1$-EPG (with $X \subseteq \{\hg,\hd,\bg,\bd\}$)
denotes the subset of $B_1$-EPG graphs where the paths can only have
shapes in $X$ (and thus $B_1$-EPG = $\{\hg\hd\bg\bd\}$-EPG).

\begin{figure}[!htbp]
\begin{center}
\includegraphics[width=0.55\textwidth]{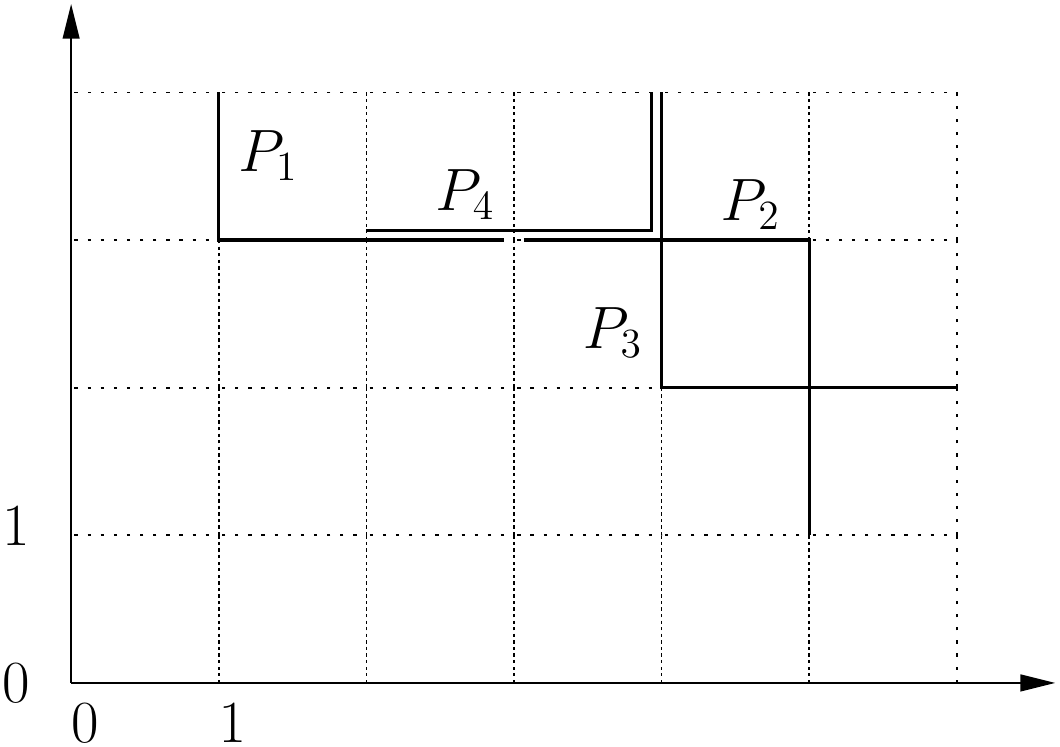}
\end{center}
\caption{B1-EPG representation of $G=(\{1,2,3,4\},\{\{1,4\},\{2,4\},\{3,4\}\})$ where
  $ver(P_1)=[3,4]$, $ver(P_1)$ lies on column $1$, $hor(P_1)=[1,3]$,
  $hor(P_1)$ lies on row $3$, and $cor(P_1)=(1,3)$.
$P_1$ and $P_2$ do not create an edge: even if $hor(P_2)=[3,5]$
  and $hor(P_2)$ also lies on row $3$, $P_1$ and $P_2$ do not share a
  common grid edge. $P_2$ and $P_3$ do not create an edge either.}
\label{fig:intro_EPG}
\end{figure}

For any path $P$ with 0 or 1 bend, we define $hor(P)$ (resp. $ver(P)$) as the interval
corresponding to the projection of $P$ on the horizontal
(resp. vertical) axis. Moreover, to describe the
position of a path we will say that $hor(P)$ (resp. $ver(P)$)
\textbf{lies} on a given row (resp. column) (see Figure~\ref{fig:intro_EPG}).
Finally, we denote by $cor(P)=(x,y)$ the coordinates of the corner of $P$.
If $P$ has no corner, let $cor(P)$ be the coordinates of some end of $P$.

\paragraph{Related work on class inclusions}
It is proved in ~\cite{edgeintersinglebend} that every graph $G$ is an EPG-graph,
and that the size of the underlying grid is polynomial in the size of
$G$.
The maximum number of bends used in the representation has been improved in~\cite{thebendnumber} where authors show that every graph of maximum degree $\Delta$ is in
$B_{\Delta}$-EPG.

Let us now consider graphs with small number of bends.  Notice first
that $B_0$-EPG graphs coincide with interval graphs.  Several recent
papers started the study EPG graphs with small number of bends.  For
example, it has been proved that $B_1$-EPG contains
trees~\cite{edgeintersinglebend}, and that $B_2$-EPG and $B_4$-EPG
respectively contain outerplanar graphs and planar
graphs~\cite{bendnumberplanar}.  We can also mention that the
recognition of $B_1$-EPG graphs is NP-hard, even when only one shape
of path is allowed~\cite{edgeinterlshaped}.

In terms of forbidden induced subgraphs, it is also 
known that $B_1$-EPG graphs exclude induced suns $S_n$ with $n \ge 4$,
$K_{3,3}$, and $K_{3,3}-e$~\cite{edgeintersinglebend}.

There is also a close relation between EPG graphs and multiple
interval graphs.  A $t$-interval is the union of $t$ disjoint
intervals in the real line. A $t$-track interval is the union of $t$
disjoint intervals on $t$ disjoint parallel lines (called tracks), one
interval on each track.  Then, a $t$-interval graph (resp. a $t$-track
graph) is the intersection graph of a set of $t$-intervals
(resp. $t$-track intervals).  Notice that $t$-track interval graphs is
a subclass of $t$-interval graphs.  It is not hard to see that
$t$-interval graphs are $B_{4(t-1)}$-EPG graphs, and that $B_{t}$-EPG
graphs are $(t+1)$-interval graphs~\cite{thebendnumber}.  Note also
that $B_1$-EPG graphs are $2$-track graphs (use one track for the rows
and one track for the columns), which in turn are $B_3$-EPG graphs
(see Figure~\ref{fig:intro2track}).

\begin{figure}[h!]
\begin{center}
\includegraphics[width=0.9\textwidth]{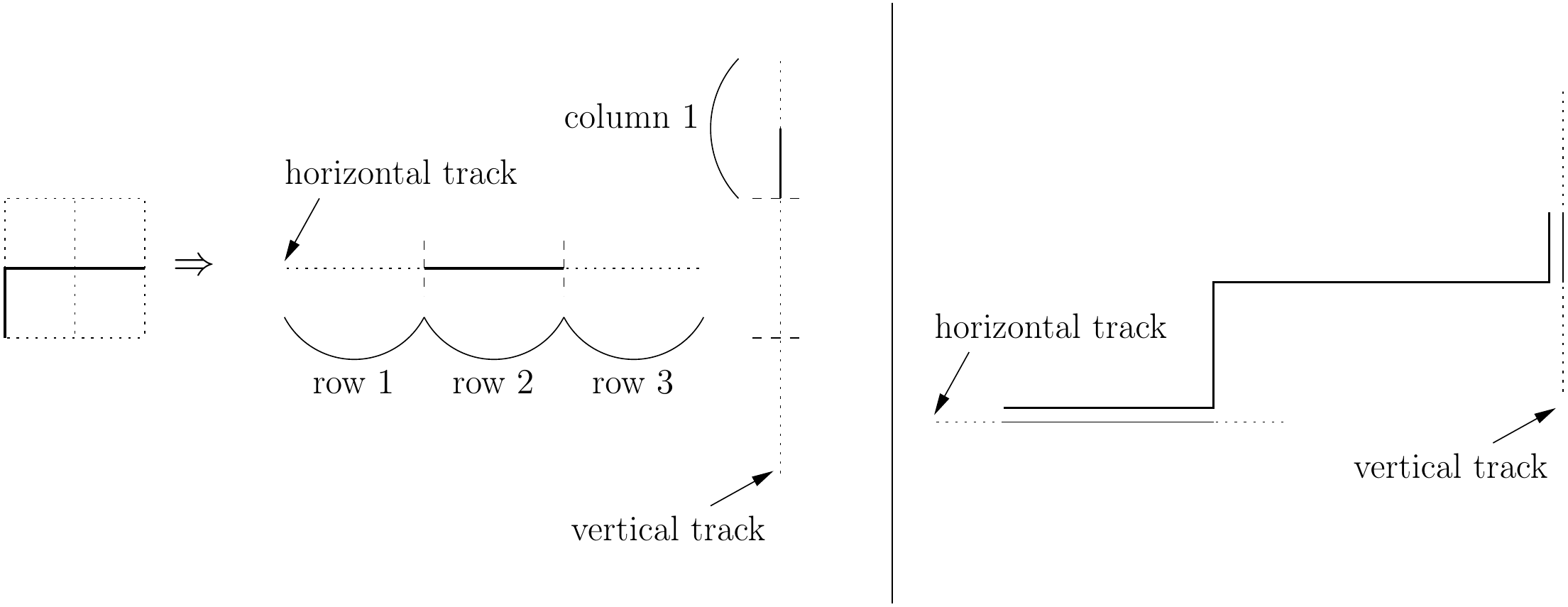}
\end{center}
\caption{
(Left) $B_1$-EPG $\subseteq$
  $2$-track: one track (horizontal here) is used for the horizontal
  parts, and the second track (vertical here) is used for the vertical
  parts. (Right) $2$-track $\subseteq$ $B_3$-EPG}
\label{fig:intro2track}
\end{figure}


\paragraph{Related work on $MIS$: Approximability.}
The main related article is~\cite{wadsMIS} where authors show that MIS
is NP-complete on $B_1$-EPG graphs, and provide a polynomial
$4$-approx-\\
imation.  As $B_1$-EPG graphs are $2$-track graphs, related
work on $MIS$ on $t$-track graph is of interest as well. MWIS (the
generalization of $MIS$ where each vertex have an arbitrary weight)
admits a $2t$-approximation on $t$-interval
graphs~\cite{schedulingsplit} (see also~\cite{FGO13}). Notice that
this answers the open problem of~\cite{wadsMIS} about finding
approximation algorithm for maximum weighted independent set on
$B_1$-EPG graphs.  It is also proved in~\cite{schedulingsplit} that
$MIS$ is APX-hard on $2$-track graphs, even when every vertex is
represented by two intervals of length $2$.  We can also
mention~\cite{approxalgointersection} that aggregates and classifies
several approximation algorithms for MIS on intersection graphs by
introducing several parameters. In particular the authors define the
notion of $k$-simplicial graphs. A graph is $k$-simplicial if and only
if there exists an order $v_1,\dots,v_n$ of the vertices such that,
for each vertex $v_i$, the subset of neighbors of $v_i$ contained in
$\{v_j | j > i\}$ can be partitioned into $k$ sets $S_1,\dots,S_k$
such that $G[S_j \bigcup \{v_i\}]$ is a clique for each $j \in
\{1,\dots,k\}$. Then they recall that MWIS is $k$-approximable in
$k$-simplicial graphs.  The $k$-simplicial graphs are related to
$B_1$-EPG graphs. Indeed the proof of the $4$-approximation of MIS on
$B_1$-EPG graphs~\cite{wadsMIS} amounts to showing that
$\{\bd\bg\}$-EPG graphs are $2$-simplicial graphs, and thus that MIS
has a $2$-approximation on $\{\bd\bg\}$-EPG graphs, and a
$4$-approximation on $B_1$-EPG graphs.  Finally, notice that the known
approximation algorithms for maximum independent set on pseudo disks
intersection graphs~\cite{maxISpseudodisks} do not directly apply
here, as two paths can cross on a grid vertex without creating an
edge.

\paragraph{Related work on $MIS$: Fixed Parameter Tractability.}
In this article we only consider the decision problem $OPT \le k$ ?
parameterized by $k$.  The main related result
is~\cite{jiang-unit2track} where the author proves that $MIS$ is
$W[1]$-complete on unit $2$-track graphs. This immediately implies that
MIS is $W[1]$-hard on $B_3$-EPG graphs.  About positive results ($FPT$
algorithm), to the best of our knowledge it seems that there is no
known $FPT$ algorithm on $B_1$-EPG or a superclass of it.

\paragraph{Contributions.}
In this article we study the approximability (in Section~\ref{sec:approx}) and the fixed parameter
tractability (in Section~\ref{sec:param}) of MIS on $B_1$-EPG.
Our main results are the following. We first show that there is no PTAS for MIS on $\{\hg\}$-EPG unless P$=$NP, even if each path has its vertical 
part or its horizontal part of length at most $3$ (\ie~ $\forall
P, (|hor(P)| \le 3 \mbox{ or } |ver(P)| \le 3))$. This improves the
NP-hardness of~\cite{wadsMIS}. Then, we show that this result
cannot be improved by showing that if $\forall P, |hor(P)| \le c$, or
if $\forall P, |ver(P)| \le c$ (where $c$ is a constant), then MIS
admits a PTAS on $B_1$-EPG graphs.
In Section~\ref{sec:param}, we show that MIS is FPT on $B_1$-EPG restricted to
only three shapes of paths, and $W_1$-hard on $B_2$-EPG. The status for
general $B_1$-EPG (with the four shapes) is left open.


\section{Approximability}\label{sec:approx}
The objective of this section is to prove that there is no PTAS for MIS on $\{\hg\}$-EPG graphs.
Let us first recall that the NP-hardness proof of MIS on $\{\hg\hd\bg\bd\}$-EPG graphs
of \cite{wadsMIS} is a reduction from MIS on planar graphs. 
As there is a PTAS for MIS on planar graphs, this is not a good candidate
to reduce from when looking for inapproximability. 
Moreover, the proof of \cite{schedulingsplit} showing
the APX-hardness of MIS on unit $2$-track graph cannot be adapted
too. Indeed, this proof consists in proving that the class of unit $2$-track graphs
(more precisely $2$-track graphs where every vertex is represented by two intervals of
length $2$) contains all graphs of degree at most $3$ (for which  MIS is APX-hard), but this class contains $K_{3,3}$ and is hence not included in $B_1$-EPG~\cite{thebendnumber}.

Our reduction is based on the classical approximation preserving
reduction from MAX-3-SAT to MIS (see for example~\cite{vazirani}, Theorem 29.13).

Let us define MAX-3-SAT(3):
\begin{itemize}
\item Input:
\begin{itemize}
\item A set of $n_{var}$ variables $\{x_i, 1 \le i \le n_{var}\}$.
\item A set of $m_{cl}$ clauses $\{C_j, 1 \le j \le m_{cl}\}$, where
  each clause is of the form $l^j_1 \vee l^j_2 \vee l^j_3$, where for
  any $1 \le t \le 3$, $\exists i$ such that $l^j_t = x_i$ or $l^j_t =
  not(x_i)$ (in both cases we say that $C_j$ contains variable $x_i$).
\item Furthermore, each variable appears at most $3$ times ($\forall
  i$, $|\{C_j| C_j \mbox{ contains } x_i\}| \le 3$). Moreover, we can
  assume that for any $i$, the positive form $x_i$ appears exactly $2$
  times, and the negative form appears exactly $1$ time.
\end{itemize}
\item Output: a truth assignment of the variables that maximizes the number of satisfied clauses 
\end{itemize}

Let us define the same function $f$ as in \cite{vazirani} (Theorem 29.13) that maps any instance $I_{sat}$ of
MAX-3-SAT(3) to a graph $f(I_{sat})=G$, where $G=(V,E)$.
For each clause $l^j_1 \vee l^j_2 \vee l^j_3$ we create a triangle $\{v^j_1,v^j_2,v^j_3\}$, and thus we have $|V|=3m_{cl}$.
For any $j$ and $t$, we say that $v^j_t$ corresponds to variable $x_i$
(resp. to the negation of variable $x_i$) if and only if $l^j_t = x_i$ (resp. $l^j_t = not(x_i)$).
For any $i, 1 \le i \le n_{var}$, we add an edge
$\{v^j_t,v^{j'}_{t'}\}$ if and only if $\exists i$ such that $l^j_t$ corresponds
to $x_i$ and $l^{j'}_{t'}$ corresponds to the negation of $x_i$.

Let $\F = \{f(I_{sat}), I_{sat} \mbox{ instance of MAX-3-SAT(3)}\}$ be the set of graphs obtained from
instances of MAX-3-SAT(3). 

\begin{proposition}[folklore]\label{prop:strictreduc}
There is a strict reduction from MAX-3-SAT(3) to MIS on graphs $\F $.
\end{proposition}
\begin{proof} 
Notice that the existence of a truth assignment satisfying $t$
clauses is equivalent to the existence of an independent set of size
$t$. As from any independent set of size $t$ we can deduce in
polynomial time a truth assignment satisfying at least $t$ clauses,
we get the desired result.
\qed\end{proof}

\begin{proposition}
\label{FinEPG}
$\F  \subseteq \{\hg\hd\}$-EPG.
\end{proposition}
\begin{proof}
Let us draw a graph $G \in \F $ using only paths of the form
$\{\hg,\hd\}$. See Figure~\ref{fig:noptasV1} for an example.
Informally, each clause $j$ corresponds to column $3j$, and each
variable $x_i$ corresponds to line $i$. Let us now define more
precisely the shapes of the paths.

\begin{figure}[h!]
\begin{center}
\includegraphics[width=0.4\textwidth]{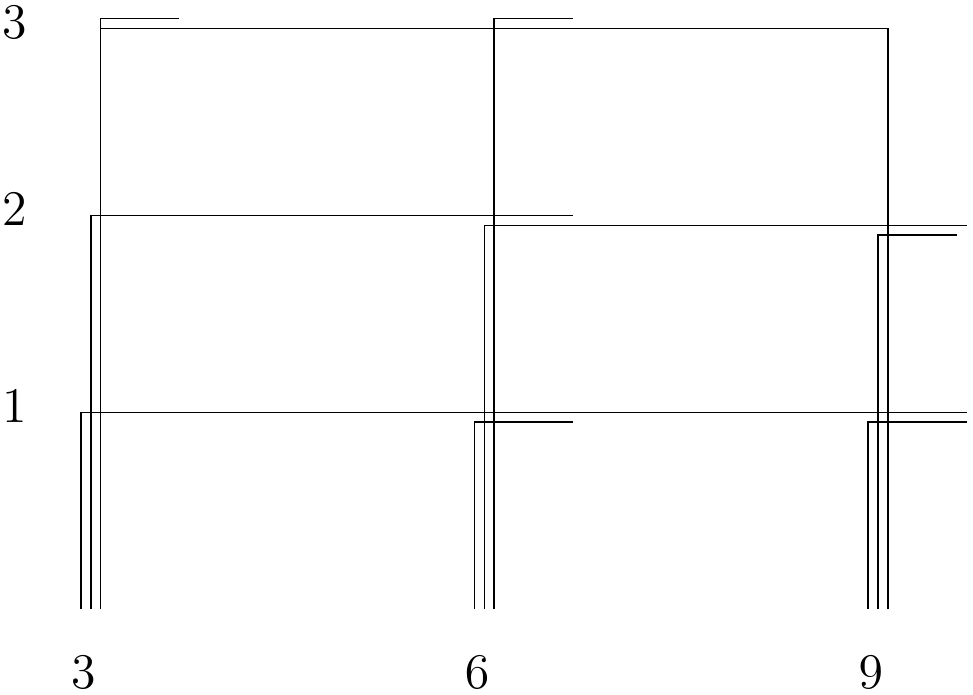}
\end{center}
\caption{Example for $C_1 = not(x_1)\vee x_2 \vee x_3$, $C_2 = x_1\vee
  not(x_2) \vee x_3$, and $C_3=x_1\vee x_2 \vee not(x_3)\}$}
\label{fig:noptasV1}
\end{figure}

Let $(V,E)$ be the vertex and edge set of $G$, with $|V|=n$.
For each vertex $v^j_t \in V$ corresponding to a variable $x_i$
or its negation, we define a path $P_{v^j_t}$, with $ver(P_{v^j_t})=[0,i]$ lying on column $3j$. 
Notice that we already have all the vertices of $V$, and the edges corresponding to the $\frac{n}{3}$ triangles.
It remains now to define the horizontal parts of the paths to encode
the adjacencies between a variable and its negation.
Let $i$ in $\{1,\dots,n_{var}\}$, and let
$\{v^{j_1}_{t_1},v^{j_2}_{t_2},v^{j_3}_{t_3}\}$ be the vertices of $G$
corresponding to variable $x_i$ or to its negation, with $j_1 < j_2 <
j_3$.
There are now three cases, according to the position of the clause containing the
negative form of $x_i$ (recall that without loss of generality we
suppose that $not(x_i)$ appears exactly one time).

Recall that the horizontal part of
$\{v^{j_1}_{t_1},v^{j_2}_{t_2},v^{j_3}_{t_3}\}$ will lie one line $i$,
and thus we just have to define the corresponding intervals.
If $C_{j_1}$ contains $not(x_i)$, then
$hor(v^{j_1}_{t_1})=[3j_1,3j_3+1]$, $hor(v^{j_2}_{t_2})=[3j_2,3j_2+1]$, $hor(v^{j_1}_{t_1})=[3j_3,3j_3+1]$.
If $C_{j_2}$ contains $not(x_i)$, then
$hor(v^{j_1}_{t_1})=[3j_1,3j_2+1]$, $hor(v^{j_2}_{t_2})=[3j_2,3j_3+1]$, $hor(v^{j_1}_{t_1})=[3j_3,3j_3+1]$.
If $C_{j_3}$ contains $not(x_i)$, then
$hor(v^{j_1}_{t_1})=[3j_1,3j_1+1]$, $hor(v^{j_2}_{t_2})=[3j_2,3j_2+1]$, $hor(v^{j_1}_{t_1})=[3j_1,3j_3]$.
This concludes the description of $G$ as an $\{\hg\hd\}$-EPG graph.

\qed\end{proof}

 Notice that this construction is not
possible from instances of MAX-3-SAT(4), as we could have for example
$C_1$ containing $x_1$, $C_2$ containing $not(x_1)$, $C_3$ containing
$x_1$, and $C_4$ containing $not(x_1)$.

As MAX-3-SAT(3) remains
MAXSNP-complete~\cite{papadimitriou2003computational} (Theorem 13.10), we get the following corollary.
\begin{corollary}
Any $(1+\epsilon)$-approximation for MIS on $\{\hg\hd\}$-EPG implies a
$(1+\epsilon)$-approximation for MAX-3-SAT(3), and thus there is no
PTAS for MIS on $\{\hg\hd\}$-EPG graphs unless P$=$NP.
\end{corollary}

Our objective is now to prove that there is no PTAS for MIS, even when
only one type of shape is allowed. Notice that the only case we need
the $\hd$ shape in the previous reduction is when $C_{j_3}$ contains
$not(x_i)$.

Let $G \in \F $, $G=(V,E)$, with $|V|=n$ and $|E|=m$. 
Let us partition $E = E_1 \bigcup E_2$,
where $E_1$ contains the $n$ edges corresponding to the $\frac{n}{3}$
triangles, and $E_2$ contains the remaining edges corresponding to edges between a variable and its negation.
To avoid using $\hd$, we will subdivide each edge of $E_2$ into $5$ edges, introducing thus $4$ new
vertices for each such edge.
More formally, let us define $G' =f'(G )$, $G'=(V',E')$. 
We start by setting $V'=V$ and $E'=E$. Moreover, for any
$e=\{v^j_t,v^{j'}_{t'}\} \in E_2$ (with $j \neq j'$), we add four new
vertices $w^e_{1}, w^e_{2}, w^e_{3}, w^e_{4}$ to $V'$, and we add
edges $\{v^j_t,w^e_1\}, \{\{w^e_i,w^e_{i+1}\}, 1 \le i \le 3\},
\{w^e_4,v^{j'}_{t'}\}$ to $E'$.
Finally, let $\F'  = \{f'(G ), G  \in \F \}$.

\begin{observation}
\label{obs1}
Let $G =(V ,E ) \in \F $, $G'  = f'(G )$.
Let $m  = |E |$.
\begin{itemize}
\item For any solution $S $ of $G $, there is a solution $S' $ of
  $G' $ such that $|S' | \ge |S |+2m $, and thus $Opt(G' ) \ge
  Opt(G )+2m $
\item For any solution $S' $ of $G' $, we can find in polynomial time
  a solution $S$ of $G$ such that $|S| \ge |S' |-2m $
\end{itemize}
\end{observation}

From the previous observation we see that solutions of $G $ and
$G' $ are simply shifted by an additive term of $2m $.
This extra term may look too large to preserve
approximability. However, the next observation shows that it only
represents a constant fraction of $Opt(G )$.

\begin{observation}
\label{obs2}
Let $G =(V ,E )=f(I_{sat}) \in \F $, with $n = |V |$, $m = |E |$,
$n_{var}$ the number of variables of $I_{sat}$, and $m_{cl}$ the
number of clauses of $I_{sat}$. As $G$ has maximum degree 4, $m\le
2n$.  By construction $n= 3m_{cl}$. For any 3-SAT instance it is known that
$\frac{7}{8}m_{cl} \le Opt(I_{sat})$. Finally, as $Opt(I_{sat}) =
Opt(G)$, one obtains that $ \frac{7}{48} m \le Opt(G)$.
\end{observation}

Let us recall the definition of AP reduction of given in~\cite{ashortguide}.
We voluntarily do not redefine here all the notations used in the
definition as they are rather explicit.
A reduction $(f,g)$ (between two optimization problems $A$ and $B$) is
said to be an AP-reduction (of parameter $\alpha$) if
\begin{enumerate}
\item For any $x \in I_A$, for any $r > 1$, $f(x,r) \in I_B$ is
  computable in time $t_f(|x|,r)$
\item For any $x \in I_a$, for any $r > 1$, and for any $y \in
  sol_B(f(x,r)), g(x,y,r) \in sol_A(x)$ is computable in time
  $t_g(|x|,|y|,r)$
\item For any fixed $r$, both $t_f(.,r)$ and $t_g(.,.,r)$ are bounded
  by a polynomial
\item For any $x \in I_A$, for any $r > 1$, and for any $y \in
  sol_B(f(x,r))$, $R_B(f(x,r),y) \le r \Rightarrow R_A(x,g(x,y,r)) \le 1+\alpha(r-1)$
\end{enumerate}

\begin{proposition}
There is an AP-reduction 
 from MIS on $\F $ to MIS on $\F' $.
\end{proposition}
\begin{proof}
Point (1), (2) and (3) of the definition of an AP reduction are
clearly verified. Let us prove (4) with $r=1+\epsilon$.
Let $S' $ be a solution of an instance $G' $ such that
$|S' | \ge \frac{Opt(G' )}{(1+\epsilon)}$.
According to Observation~\ref{obs1}, we get a solution $S$ of size
at least $|S' |-2m =\frac{Opt(G' )}{1+\epsilon}-2m  \ge
\frac{Opt(G )}{1+\epsilon}-2m (\frac{\epsilon}{1+\epsilon})$.
Using Observation~\ref{obs2}, we deduce $|S| \ge
\frac{Opt(G )}{1+\epsilon}(1-2\epsilon\frac{48}{7})$, 
 which concludes the proof.
\qed\end{proof}

\begin{proposition}\label{prop:F}
$\F'  \subseteq \{\hg\}$-EPG.
\end{proposition}
\begin{proof}
We will adapt the structure used in Proposition~\ref{FinEPG}. 
Let $G' =(V',E') \in \F' $.
For each vertex $v^{j}_t \in V'$ corresponding to a variable
$x_i$ or its negation, we define a path $P_{v^j_t}$, with $ver(P_{v^j_t})=[0,3i]$ lying on column $3j$. 
As before, we already have the vertices and the edges corresponding to the triangles of $V'$.
It remains now to add the vertices and edges corresponding to the
subdivided edges.
Let $\{v^{j_1}_{t_1},v^{j_2}_{t_2},v^{j_3}_{t_3}\}$ be the vertices of $G$
corresponding to variable $x_i$ or to its negation, with $j_1 < j_2 <
j_3$.
There are now three cases, according to the position of the clause containing the
negative form of $x_i$ (recall that without loss of generality we
suppose that $not(x_i)$ appears exactly one time).
Thus, we create the appropriate paths as depicted in Figure~\ref{fig:noptasV2}.
Notice that now, for each variable $i$ in $\{1,\dots,n_{var}\}$, the required
subdivided edges will be created using only lines $3i$, $3i+1$ and $3i+2$ of
the underlying grid.
\begin{figure}[h!]
\begin{center}
\includegraphics[width=1\textwidth]{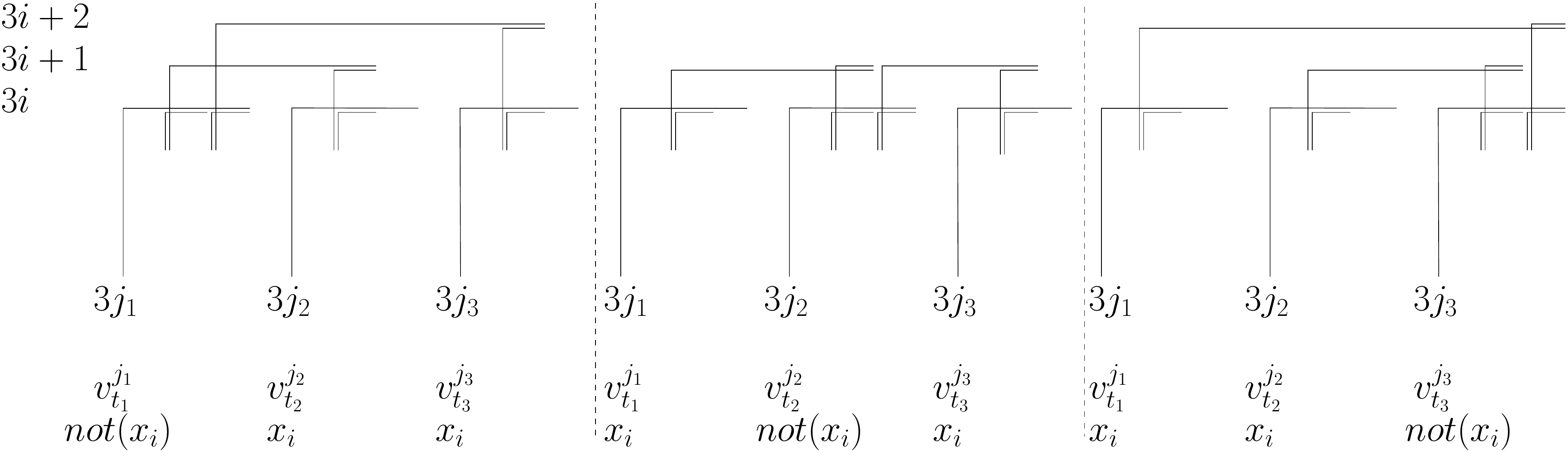}
\end{center}
\caption{The three possible cases according to the position of $not(x_i)$}
\label{fig:noptasV2}
\end{figure}
\qed\end{proof}

We are ready to state the main inapproximability result, whose
proof is now immediate. 
\begin{theorem}
\label{thm:noptas}
There is no PTAS for MIS on $\{\hg\}$-EPG unless P$=$NP, even if each path has its vertical 
part or its horizontal part of length at most $3$ (\ie~ $\forall
P, (|hor(P)| \le 3 \mbox{ or } |ver(P)| \le 3))$.
\end{theorem}

As MIS is APX-hard on 2-track graphs~\cite{schedulingsplit}, even when every vertex is represented by two intervals of
length $2$, it is natural to ask the same question here, \ie~
to determine if Theorem~\ref{thm:noptas} can be extended to paths whose vertical
and horizontal intervals have constant length.
Let us first notice that this restriction remains NP-hard.

\begin{proposition}\label{prop:NPhard}
MIS remains NP-hard on $\{\hg\}$-EPG graphs, even if all the paths have their
horizontal part and their vertical part of length at most 2 (\ie~ $\forall
P, |hor(P)| \le 2$ and $|ver(P)| \le 2$). 
\end{proposition}
\begin{proof}
We will only sketch the proof as it is a straightforward modification
of the NP-hardness proof of \cite{wadsMIS}. The reduction of
\cite{wadsMIS} is from MIS on a planar graph $G$ with maximum degree
four.  They first recall that we can construct in polynomial time an
embedding in a grid such that edges of $G$ are piece-wise linear
curves following the grid lines (as depicted
Figure~\ref{fig:reducfromplanar} on the left). Then, they subdivide
each edge of $G$ (by adding $a_e$ vertices inside edge $e$ , with
$a_e$ an even integer) to get another planar graph $G'$, and observe
that $Opt(G')=Opt(G)+\sum_{e}\frac{a_e}{2}$. Finally, they draw graph
$G'$ as a $B_1$-EPG graph.

\begin{figure}[h!]
\begin{center}
\includegraphics[width=1\textwidth]{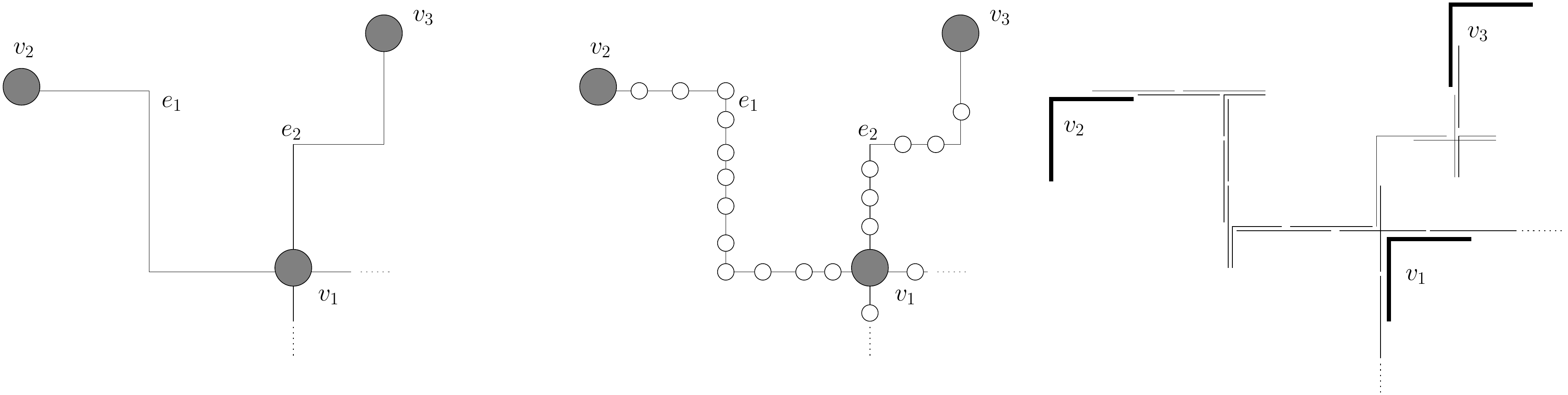}
\end{center}
\caption{Reduction from planar graph using only $\hg$ shapes whose
  both direction are bounded. Here we have $a_{e_1}=12$ and $a_{e_2}=6$.}
\label{fig:reducfromplanar}
\end{figure}

Thus, it is sufficient to notice that by subdividing even more each
edge we can get a graph $G''$ such that
$Opt(G'')=Opt(G)+\sum_{e}\frac{a'_e}{2}$ (with $a'_e \ge a_e$), and
that we can draw using only $\hg$ shapes of horizontal and vertical
parts of length at most $2$ (see Figure~\ref{fig:reducfromplanar}).
Notice that for a vertex of degree 4, the associated path necessarily
has total length at least 4, as its neighbors are pairwise non
adjacent.
\qed\end{proof}

As shown in the following theorem, it is not possible to improve
the inapproximability, even if only one direction
(but always the same) is bounded, and even if the four
shapes are allowed.

\begin{theorem}
\label{thm:noptas2}
MIS admits a PTAS on $B_1$-EPG graphs where all the paths have a
horizontal part of length at most $c$, where $c$ is a constant (\ie~
$\forall P, |hor(P)| \le c$). More precisely, we can find an independent set of size 
at least $(1-\epsilon)|OPT|$ in $\mathcal{O}^*(n^{3c\frac{1}{\epsilon}})$ time.
\end{theorem}

We will prove Theorem~\ref{thm:noptas2} using the classical Baker shifting
technique~\cite{Bak94}.
We could expect that such a PTAS is already known, as MIS has been
widely studied on intersection graphs (see for example the PTAS of~\cite{maxISpseudodisks}
for MIS on intersection of Pseudo-Disks).
However, as two paths can cross without creating an edge in the EPG
model, 
this requires an ad-hoc adaptation of the shifting technique.

\begin{proof}
Let $G=(V,E)$ be a $B_1$-EPG graph where all the paths have
an horizontal part of length at most $c$.
Without loss of generality, let us suppose that all the paths are
drawn in the positive quadrant of the plane (\ie~with positive
coordinates), and that the underlying grid is a square of size $s \le
poly(n)$, where $n = |V|$.

Let us consider a given optimal solution $OPT$.
Let $k \in \N^*$. Our goal is to get a solution of size at least
$|OPT|(1-\frac{1}{k})$.

For any integers $i\in \N$, let $X_i$ be the set of paths $P$ of $G$ such that $[i,i+1]\subseteq hor(P)$.
For any $d \in \{0,\dots,kc-1\}$, let $R_d = \bigcup_{a \in \N}X_{d+akc}$ (see Figure~\ref{fig:baker}).
Let $OPT_d = OPT \bigcap R_d$, \ie~$OPT_d$ is the set of all paths of
$OPT$ which cross one of the vertical strips, between column $d+akc$ and $d+akc+1$ for some $a\in \N$.

\begin{figure}[h!]
\begin{center}
\includegraphics[width=0.7\textwidth]{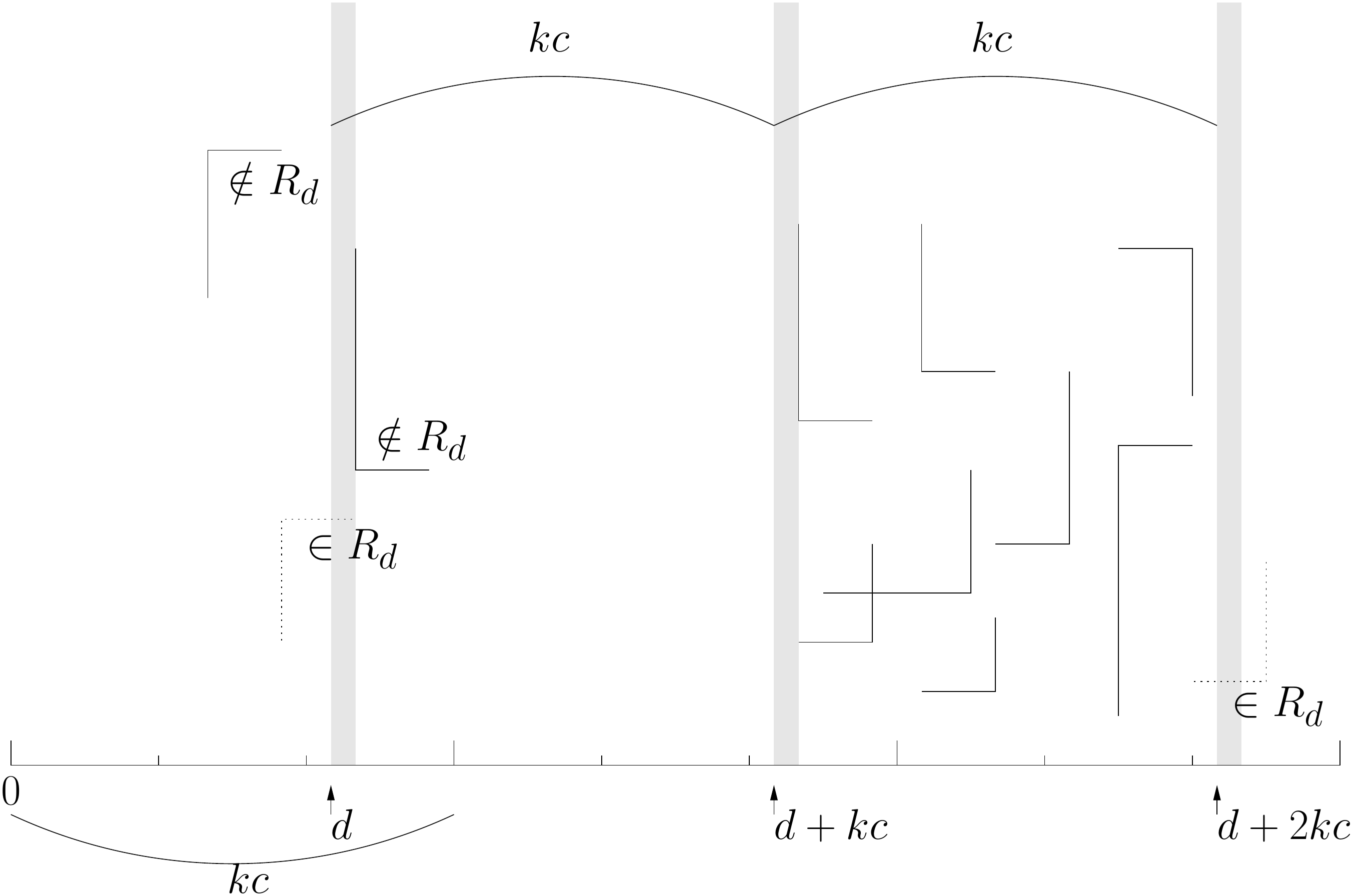}
\end{center}
\caption{$R_d$ is the set of paths crossing one of the gray strips. We solve MIS exactly in all the
  white strips, \ie~between column $d+ack+1$ and $d+(a+1)ck$ with $a \ge -1$.}
\label{fig:baker}
\end{figure}

Let $d_0 = min_{d}|OPT_d|$.
Observe first that $|OPT_{d_0}| \le \frac{1}{k}|OPT|$.
Indeed, $\sum_{d=0}^{kc-1}|OPT_d| \le c |OPT|$, as each vertex $v^*$ of $OPT$
belongs to at most $c$ different $R_l$, 
and thus $kc|OPT_{d_0}| \le c |OPT|$.

Thus, for each $d$ the algorithm solves optimally the problem on
$G\setminus R_d$ and gets a solution $A_d$. Then, it returns the
largest solution among the $A_d$.
Let us now see how we solve MIS optimally on $G'=G\setminus R_d$.

Observe that removing the vertices of $R_d$
disconnects the graph. Indeed, paths at the left of the strip (\ie~$P$
such that $hor(P)\subseteq [0, d+akc]$) cannot
be connected to paths on the right of the strip (\ie~$P$
such that $hor(P)\subseteq [d+akc+1, +\infty]$).
Thus, each connected component of $G'$ correspond to a $B_1$-EPG graph defined on a grid
with $\beta=kc$ columns. Naturally, we compute an optimal
solution on $G'$ by taking the union of optimal solutions of each
connected component.

It remains now to solve MIS on a $B_1$-EPG graph where the underlying
grid has a constant number $\beta$ of columns.
To that end we write a dynamic programming algorithm that parses the
input from bottom to top, and only remembers for each column the highest point already
occupied by a path.
More precisely, let us design an algorithm $DP(l_{min},y_1,\dots,y_{\beta})$
that computes a maximum independent set of paths $S$ with the extra
constraints that
\begin{itemize}
\item for any $P \in S$, $hor(P)$ lies on a row $r \ge l_{min}$
\item for any $P \in S$, for any $1 \le i \le \beta$, if $ver(P)=[t_1,t_2]$ (\ie~the lowest point of the vertical part is on row $t_1$) and lies on column $i$, then $t_1 > y_i$ (\ie~for all $i$, all the
  vertical parts lying on column $i$ must be above $y_i$)
\end{itemize}
To compute such a set $S$, we guess
in an optimal solution (satisfying the extra constraints) the index
$l^* \ge l_{min}$ of the lowest row used, and we also guess all the
paths of this optimal solution whose
horizontal part lie on $l^*$. More precisely  (see Figure~\ref{fig:progdyn}),
for any $l \ge l_{min}$, for any subset of paths $S'$ such that
\begin{itemize}
\item $S'$ is an independent set
\item for any $P \in S'$, $hor(P)$ lies on a row $l$
\item for any $P \in S'$, for any $1 \le i \le \beta$, if $ver(P)=[t_1,t_2]$ and lies on column $i$, then $t_1 > y_i$
\end{itemize}
$DP(l_{min},y_1,\dots,y_{\beta})$ considers the solutions of type $S'
\bigcup DP(l+1,y'_1,\dots,y'_{\beta})$, and keep the best.  Notice
that a path that has no horizontal part and a vertical part of the
form $[l,l+d]$ will be considered in set $S'$ (\ie~in such a case its
horizontal part is assumed to be on row $l$).  It remains to define
the $y_i'$. If no path of $S'$ has its vertical part lying on column
$i$, then $y_i' = y_i$. Otherwise, $y_i' = max\{t_2| \exists P \in S'$
such that $ver(P)=[t_1,t_2]$ and lies on column $i$
(notice that there can be two paths lying on the same column, as
depicted Figure~\ref{fig:progdyn}).

\begin{figure}[h!]
\begin{center}
\includegraphics[width=0.9\textwidth]{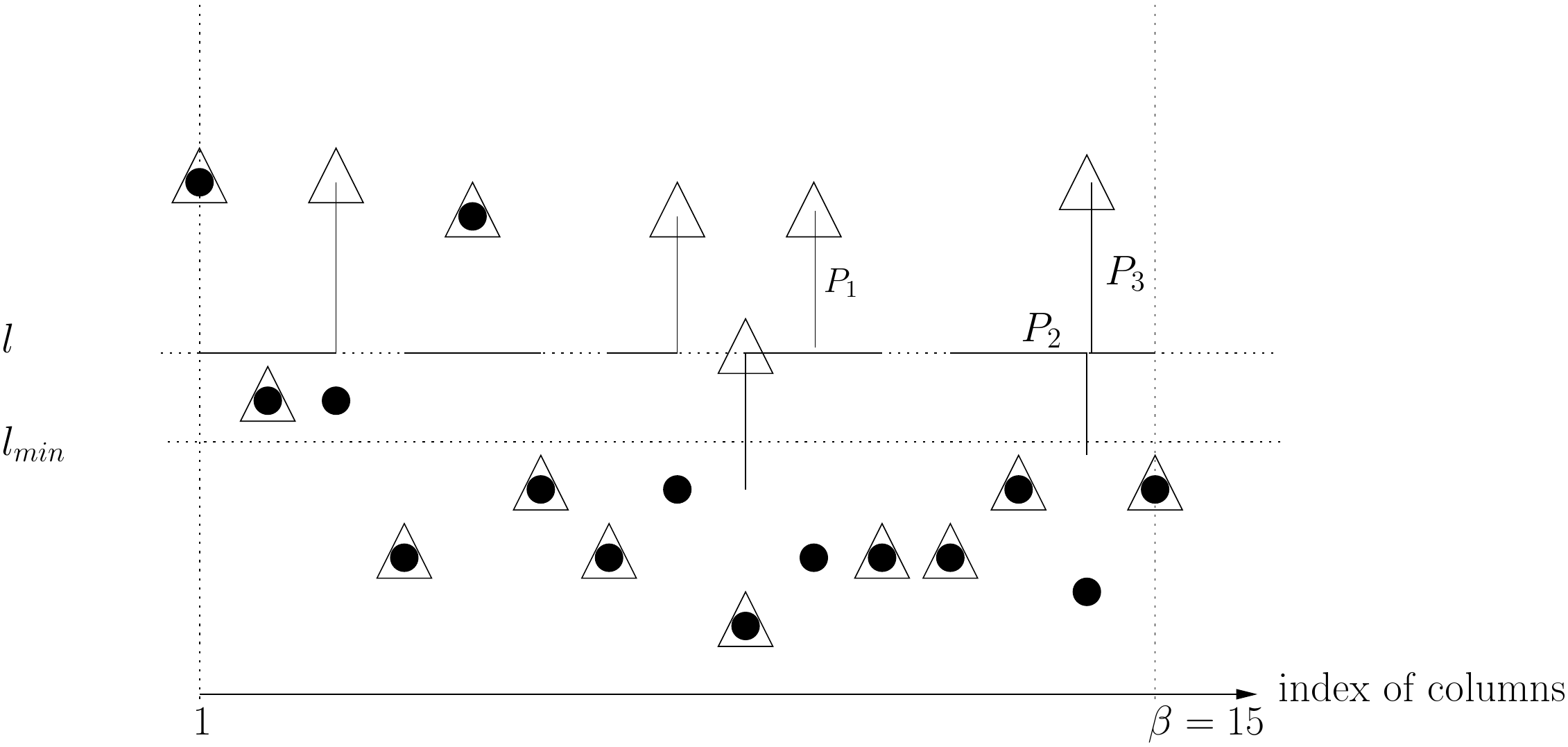}
\end{center}
\caption{Example of one of the solutions enumerated by
  $DP(l_{min},y_1,\dots,y_{\beta})$. The $y_i$ values are represented
  by black circles, and the new $y_i'$ values are represented
  by triangles. The depicted set $S'$ contains $7$ paths, whose
  horizontal parts lie on row $l$ (we consider that $P_1$ has an
  horizontal part lying on row $l$, even if it has no horizontal
part).
Observe that both $P_2$ and $P_3$ have a vertical part lying on column
$\beta-1$, and thus $y'_{\beta-1}$ is set according to the highest
path ($P_3$ here).
Observe also that $y_i$ is not necessarily lower than $l_{min}$.}
\label{fig:progdyn}
\end{figure}

Let us now bound the complexity of $DP$. As $l_{min}$ and any $y_i$ lies in
$\{1,\dots,n\}$, the total number of input is bounded by
$n^{\beta+1}$. Concerning the number of branches examined by the
algorithm, we note that there are at most $n$ possible values for
$l$, and that $|S'| \le 2\beta-1$ (in the case where
$S'$ contains one vertical path on every
column and without horizontal part, and $\beta-1$ horizontal paths
of length $1$ on row $l$ and without vertical part). Thus, there is at most $n^{2\beta-1}$
branches, leading to an overall complexity in
$\mathcal{O}^*(n^{3\beta})$.  As $\beta = kc$ and the
approximation ratio is $1-\frac{1}{k}$, this provides the desired
complexity for the PTAS and concludes the proof of
Theorem~\ref{thm:noptas2}.



\qed\end{proof}





\section{Fixed parameter tractability}\label{sec:param}
In this section we consider the MIS problem in the standard
parameterization, and thus we consider the decision problem $OPT \le k$ ? parameterized by $k$.  
Recall first that as $B_1$-EPG $\subseteq$ $2$-track $\subseteq$
$B_3$-EPG and as it is proved in ~\cite{jiang-unit2track} that MIS is $W_1$-hard in
unit $2$-track graphs, we already know that MIS is $W_1$-hard in
$B_3$-EPG graphs. 

In this section we prove that MIS is FPT on $B_1$-EPG restricted to
only three shapes of paths, and $W_1$-hard on $B_2$-EPG graphs. The status for
general $B_1$-EPG (with the four shapes) is left open. 

\subsection{FPT algorithm for MIS on $B_1$-EPG with three shapes}
The principle of our FPT algorithm
is to repeat the following process: locate a set of $17k^2$ paths of the
instance which contains an element of an optimal solution, and then
branch on these $17k^2$. If the instance is positive (\ie has a MIS of
size at most $k$) then the algorithm will find a solution after at
most $17k^{2k}$ choices. Before detailing the algorithm, we need the
following definitions and lemmas which will useful to find an element
of an optimal solution.

We refer the reader to Figure~\ref{fig:FPTgroupe} for the next
definitions.  Notice that a purely vertical path $P$ with
$ver(P)=[a,b]$ lying on column $c$ is considered as a path of shape
$\bg$ with $hor(P)=[c,c]$ lying on row $a$.  In the same way, a purely
horizontal path $P$ with $hor(P)=[a,b]$ lying on row $r$ is considered
as a path of shape $\bg$
with $ver(P)=[r,r]$ lying on column $a$. \\
Notice also that if $G$ is a $B_1$-EPG with a path $P$ entirely
containing another path $P'$ then $Opt(G)=Opt(G\setminus \{P\})$ for
the MIS problem. Indeed, if an independent set $I$ of $G$ contains $P$
then $(I\setminus \{P\})\cup \{P'\}$ is also an independent set of $G$
with the same size than $I$. Thus in every $B_1$-EPG instance of MIS
we will consider, we assume that there is no path entirely contained in
another path.\\

  A subset $G_*$ of paths is a \textbf{group} if and only if all paths
  of $G_*$ have the same shape and the same corner (ie. $\forall \
  P_1,P_2 \in G_*\ cor(P_1) = cor(P_2)$).

  Let $G_*$ be a group of paths of shape $\bg$.  As we can remove from
  the instance any path $P$ that contains entirely another path $P'$,
  we have for any $P, P' \in G_*$, $hor(P) \subseteq hor(P') \Rightarrow
  ver(P') \subseteq ver(P)$.  We say that \textbf{$P \in G_*$ is the
    rightmost path of $G_*$} (resp. topmost) if and only if for every
  $P' \in G_*$ we have $hor(P') \subseteq hor(P)$ (resp. $ver(P') \subseteq
  ver(P)$).
  We also adapt the definition for the three other shapes of groups
  (for example for a group of shape $\hd$, we define the leftmost and
  the downmost path).

  Let $S$ be an independent set of a $B_1$-EPG.  Let $x \in S$ with
  $hor(x)=[a,b]$, lying on row $r$.  We say that \textbf{the right
    side of $x$ is free in $S$} (or simply that \textbf{the right side
    of $x$ is free}) if there is no $y \in S$, $y \neq x$, lying on
  row $r$ such that $hor(y)=[c,d]$ and $c\ge b$.
  We define in the same way that the left, up, and down side of $x$ is
  free in $S$.

The \textbf{main sides} of a path are the ones given by its shape:
a path of shape $\bg$ (resp. $\hg$, $\hd$ or $\bd$) has main sides up and right
(resp. down and right, down and left or up and left).\\

The two following lemmas provide conditions allowing us to guess in FPT time a
vertex of an optimal solution.
\begin{lemma}
\label{lemmeFPT1}
Let $r$ be a fixed row, and let $X$ be a subset of paths such
that every $P \in X$ has shape $\bg$ and $hor(P)$ lies on row $r$.
We can construct in polynomial time a set $f_1(X)=X'\subseteq X$ with $|X'| \le k$ verifying the following property.

If there exists an independent set $S$, $|S|=k$ and a $x^* \in S \cap X$ such that the
right side of $x^*$ is free in $S$, then there exists $S'$, $|S'| = k$ such that $S' \cap X' \neq \emptyset$.
\end{lemma}

\begin{proof}
  Let $\{a_i\ :\ 1 \le i \le i_{max}\}=\{j\ :\ \exists P \in X$ with
  $ver(P)$ lying on column $j\}$ be the set of columns used by the
  vertical parts of paths of $X$. Let us suppose without loss of
  generality that $a_i > a_{i+1}$ for any $i$.  We partition $X$ into
  (maximal) groups $G_i$ for $1 \le i \le i_{max}$: we define $G_i$ as
  the set of all paths of $X$ having their corner at coordinates
  $c_i=(a_i,r)$ (see Figure~\ref{fig:FPTgroupe}).  We denote by $i^*$
  the index of the group of $x^*$ (ie.  $x^* \in G_{i^*}$).  Let
  $X'=\{P'_1,\dots,P'_{min(i_{max},k)}\}$, where $P'_i$ is the
  rightmost path of $G_{i}$. Finally, let $hor(P'_i)=[a_i,b_i]$ and
  $hor(x^*)=[a^*,b^*]$.

\begin{figure}[h!]
\begin{center}
\includegraphics[width=0.9\textwidth]{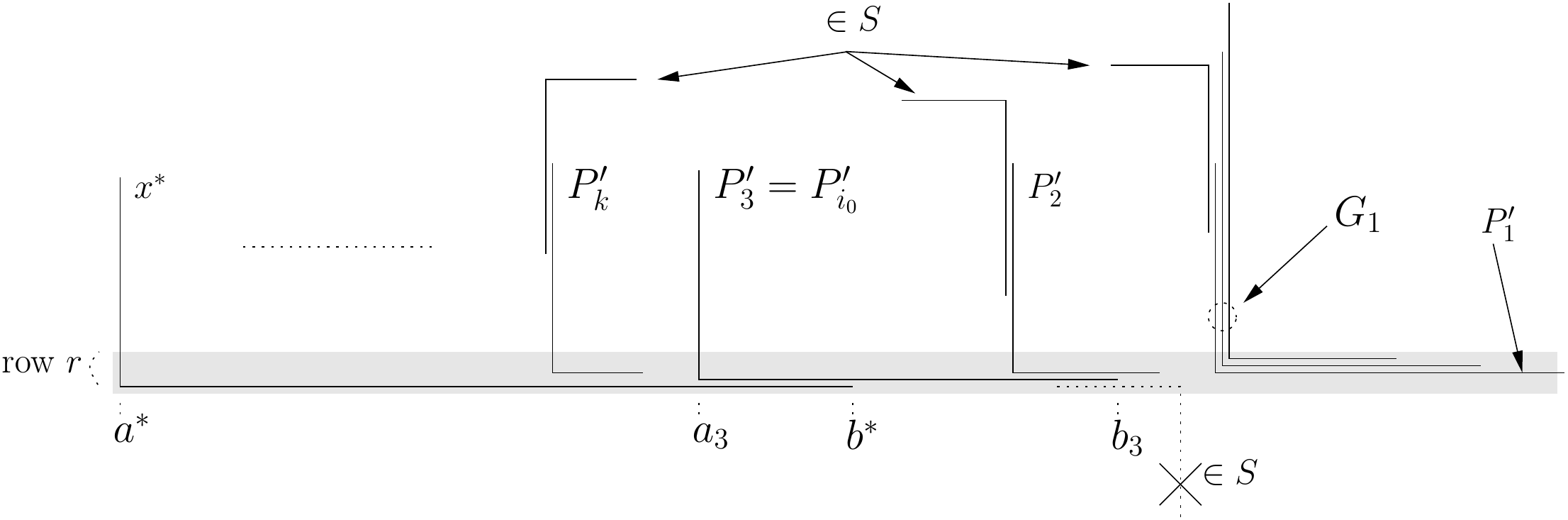}
\end{center}
\caption{Example of a group and of the case where the right of $x^*$ is free, and $x^* \in G_{i^*}$ with $i^* > k$ in Lemma~\ref{lemmeFPT1}. }
\label{fig:FPTgroupe}
\end{figure}

Suppose first that $i^* \le k$.
As the right side of $x^*$ is free in $S$, we can remove $x^*$ from
$S$ and chose $P'_{i^*}$ to get another solution $S'$ of size $k$.

Suppose now that $i^* > k$.  Notice first that in this case by
construction of the groups, we have $a^* < a_i$ for any $i \le k$ (as
we collected first paths with rightmost left endpoint). If there
exists $x\in S\setminus \{x^*\}$ with $hor(x)$ lying on row $r$,
then we have $hor(x)=[a,b]$ with $b<a^*$ (because the right side of
$x^*$ is free in $OPT$). So we have $b<a_i$ for any $i\le k$ and $x$
does not intersect paths of $X'$. Equivalently paths of $X'$ can only
intersect $S \setminus \{x^*\}$ by sharing vertical edges of the
underlying grid.
Thus as $|S|= k$, there exists $i_0 \le
k$ such that $(S \setminus \{x^*\}) \cup \{P'_{i_0}\}$ is an
independent set (\ie~we swap $x^*$ and $P'_{i_0}$ to get $S'$),
and $X'$ satisfies the claimed property.

\qed\end{proof}


\begin{lemma}
\label{lemmeFPT2}
Let $r$ be a fixed row, and let $X$ be a subset of paths such
that every $P \in X$ has shape $\bd$ and $hor(P)$ lies on row $r$.
We can construct in polynomial time a set $f_2(X)=X'\subseteq X$ with $|X'| \le k$ verifying the following property.

If there exists an independent set $S$, $|S|=k$ and a $x^* \in S \cap X$ such that the
right side \textbf{and} the top side of $x^*$ are free in $S$, then there exists $S'$, $|S'| = k$ such that $S' \cap X' \neq \emptyset$.
\end{lemma}
\begin{proof}
  Let us partition $X$ into groups, but not as in the previous lemma.
  Let $X=\{P_1,\dots,P_{|X|}\}$ be sorted such that $a_i \ge a_{i+1}$
  where $hor(P_i)=[a_i,b_i]$.  We create the groups according to the
  following procedure. We start with $G_1 = \{P_1\}$ and $i_{max}=1$,
  and we say that $G_1$ is created. Then, for $i$ from $2$ to $|X|$,
  if the corner of $P_i$ is equal to the corner of an already created
  group $j \le i_{max}$, then we add $P_i$ to $G_j$. Otherwise, we set
  $i_{max}=i_{max}+1$ and create a new group $G_{i_{max}} = \{P_i\}$.
  We denote by $i^*$ the index of the group of $x^*$ (ie.  $x^* \in
  G_{i^*}$).  Let $X'=\{P'_1,\dots,P'_{min(i_{max},k)}\}$, where
  $P'_i$ is the topmost path of $G_{i}$.  Finally, let
  $hor(P'_i)=[a'_i,b'_i]$, and $hor(x^*)=[a^*,b^*]$.

  Suppose first that $i^* \le k$. In this case, as the top side of
  $x^*$ is free in $S$, we can remove $x^*$ from $S$ and chose
  $P'_{i^*}$ to get another solution $S'$ of size $k$.

  Suppose now that $i^* > k$.  Notice first that in this case, by
  construction of the groups we have $a^* < a'_i$ for any $i\le k$ (as
  we collected first paths with rightmost left endpoint).  The proof
  is now exactly the same as in Lemma~\ref{lemmeFPT1}: paths of $X'$
  can only intersect $S \setminus \{x^*\}$ by sharing vertical edges
  of the underlying grid, and thus there exists $i_0 \le k$ such that
  $(S \setminus \{x^*\}) \cup \{P'_{i_0}\}$ is an independent set
  (\ie~we swap $x^*$ and $P'_{i_0}$ to get $S'$).
\qed\end{proof}



\begin{lemma}
\label{lemmeFPT3}
Let $X$ be a subset of paths such that every $P \in X$ has shape $\bg$.
We can construct in polynomial time a set $f_3(X)=X'\subseteq X$ with $|X'| \le 4k^2$ verifying the following property.

If there exists an independent set $S$, $|S|=k$ and a $x^* \in S \cap X$ such that the two orthogonal sides of $x^*$ are free
in $S$, and one of this side is a main side of $x^*$ (\ie~the free sides of $x^*$ are right/up, right/down or up/left),
 then there exists $S'$, $|S'| = k$ such that $S' \cap X' \neq \emptyset$.
\end{lemma}
\begin{proof}
  Let $x^* \in S$.  We first construct greedily, in polynomial time, a maximal independent set
  $A=\{P_1,\dots,P_{|A|}\}$, $|A| \le k$ (if $|A| > k$ we define $X'$ as $k$ arbitrary vertices of $A$). Let $S_i^{ver}$ be the set
  of paths $P$ in the input with shape $\bg$ and which intersects
  $P_i$ by sharing a vertical edge of the underlying grid. Let also
  $S_i^{hor}$ be defined in the same way for horizontal edges (notice
  that we may have $S_i^{ver} \cap S_i^{hor} \neq \emptyset$).  As
  $A$ is maximal, there exists $i_0$ such that $x^*$ intersects
  $P_{i_0}$, and thus such that $x^* \in S_{i_0}^{ver}$ or $x^* \in
  S_{i_0}^{hor}$.  Let us suppose first that $x^* \in S_{i_0}^{hor}$.
  If the right side of $x^*$ is free in $S$, then by
  Lemma~\ref{lemmeFPT1} we know that we can replace $x^*$ by a $x^{*'} \in
  f_1(S_{i_0}^{hor})$. Otherwise, the topside and the leftside of
  $x^*$ are free in $S$. In this case we know by
  Lemma~\ref{lemmeFPT2} that we can replace $x^*$ by a $x^{*'} \in
  f_2(S_{i_0}^{hor})$.\\
  Considering also the case $x^* \in S_{i_0}^{ver}$, we know that
  $x^{*'}$ can be chosen in $X'=\bigcup_{i=1}^{|A|} (f_1(S_{i}^{hor})
  \cup f_2(S_{i}^{hor}) \cup f_1(S_{i}^{ver}) \cup f_2(S_{i}^{ver}))$
  which of size at most $4k^2$.  \qed\end{proof}

\begin{lemma}
\label{lemmeFPT4}
Let $(G,k)$ be a instance of MIS on $\{\hg\bg\bd\}$-EPG graphs.
We can construct in polynomial time a set $X'$ with $|X'| \le 17k^2$ verifying the following property.

If there exists an independent set $S$, $|S|=k$, then there exists $S'$, $|S'| = k$ such that $S' \cap X' \neq \emptyset$.
\end{lemma}

\begin{proof}
   Let $x^* \in S$ (with
  $cor(x^*) = (a^*,b^*)$) be the top-right most path of $S$ (\ie~for
  any $x \in S$ with $cor(x)=(a,b)$, either $b < b^*$ (Case 1) or $b
  = b^*$ and $a \le a^*$ (Case 2).  Notice that the only case where $b
  = b^*$ and $a = a^*$ is when $x^*$ has shape $\bd$ and there exists
  a unique $x \in S$ with shape $\hg$, and $cor(x)=cor(x^*)$.

  (Case 1) Let us first consider the case where for any $x \in S$
  with $cor(x)=(a,b)$, either $b < b^*$ or ($b = b^*$ and $a <
  a^*$).  Whatever the shape of $x^*$, $x^*$ has always two orthogonal free sides in $S$,
  and one of this side is a main side (for instance, if $x^*$ has
  shape $\bg$ then his two main directions are free, if $x^*$ has shape
  $\hg$ or $\bd$, one of his main direction and another orthogonal
  direction are free). Thus, we apply three times Lemma~\ref{lemmeFPT3} (one time for $X$ containing all the paths of a given shape), 
  and we construct in polynomial time a set $Y$ of size at most $12k^2$ verifying the desired property.

  (Case 2) It remains now the case where $x^*$ has shape $\bd$ and
  there exists a unique $x\in S$ with shape $\hg$, and
  $cor(x)=cor(x^*)$.  In this case, as in Lemma~\ref{lemmeFPT3} we
  first construct in polynomial time a maximal independent set
  $A=\{P_1,\dots,P_{|A|}\}$, and thus (re-using the notation of
  Lemma~\ref{lemmeFPT3}) we know that there exist $i_0, i_1$ such that
  ($x^* \in S_{i_0}^{ver}$ or $x^* \in S_{i_0}^{hor}$) and
  ($x' \in S_{i_1}^{ver}$ or $x' \in S_{i_1}^{hor}$).\\
  If $x^* \in S_{i_0}^{ver}$, then by Lemma~\ref{lemmeFPT1} we know
  that we can chose $Z=f_1(S_{i_0}^{ver})$ (of size at most $k$)
  In the same way, if $x \in S_{i_1}^{hor}$, then by
  Lemma~\ref{lemmeFPT1} we know that we can chose $Z=f_1(S_{i_1}^{hor})$. So in these cases, we know that we
  can chose $Z=\bigcup_{i=1}^{|A|} (f_1(S_{i}^{hor}) \cup f_2(S_{i}^{hor}) \cup
  f_1(S_{i}^{ver}) \cup f_2(S_{i}^{ver}))$ of size at most $4k^2$.
  It remains the case where $x^* \in S_{i_0}^{hor}$ and $x \in
  S_{i_1}^{ver}$. Let $c_{i_1}$ be the column of the vertical part of
  $P_{i_1}$, and $r_{i_0}$ be the row of the horizontal part of
  $P_{i_0}$. In this last case we can deduce that
  $cor(x^*)=cor(x)=(c_{i_1},r_{i_0})$. Thus, we can replace in $S$
  vertex $x^*$ by $x_{i_0,i_1}$: the topmost path of the group
  $G_{i_0,i_1}$ of $\bd$ paths having their corner at
  $(c_{i_1},r_{i_0})$. So in this case, we find an element of an
  optimal solution belongs the set of the topmost paths of each group
  $G_{i,j}$ with $1 \le i,j \le k$. This set has size $k^2$.

  To summarize, if we want to ensure that we capture a vertex of another solution $S'$ of the same size,
  we have to construct the set $Y$ of size $12k^2$
  for the first case, and a set $Z$ of size $4k^2+k^2=5k^2$ for the second
  case, leading to a choice between at most $17k^2$ paths.
\qed\end{proof}

We are now ready to prove the main result of this section whose proof is immediate using Lemma~\ref{lemmeFPT4}.
\begin{theorem}\label{thm-FPT}
  The question $OPT \ge k$ can be solved in time $O(k^{2k}poly(n))$ in
  $B_1$-EPG graph with three shapes of paths.
\end{theorem}

\subsection{W[1]-hardness of MIS on $B_2$-EPG graphs}

\begin{theorem}\label{thm-W1}
MIS is W[1]-hard, even restricted to $B_2$-EPG graphs.
\end{theorem}
\begin{proof}
We proceed by an FPT-reduction from the $k$-MULTICOLORED CLIQUE
problem introduced in~\cite{k-CLIQUE}, and used there to prove the
W[1]-hardness of MIS for unit 2-interval graphs. The reduction we
present here is inspired on this reduction and on the one used for
2-track graphs~\cite{jiang-unit2track}. Let $k$ be any positive
integer, and consider a graph $G$ with vertex set $V=V_1\cup\ldots\cup
V_k$, where each set $V_i$ is a stable set.  We are going to construct
a $B_2$-EPG graph $H$ such that $G$ has a $k$-clique if and only if
$H$ has an independent set of size $3k^2$.

As the graph $H$ will be such that its vertex set decomposes into
$6{k\choose 2} +3k = 3k^2$ induced cliques, $H$ cannot have an independent
set of size greater than $3k^2$. Furthermore, $H$ will
have an independent set of size $3k^2$ (\ie~$H$ is a
positive instance) if and only if one can select one vertex from each
of these cliques.

Let us define an {\bf edge gadget}. For any pair $i,j$ of integers
such that $ 1\le i<j\le k$, we construct a gadget in $H$ that
corresponds to the edges between the sets $V_i$ and $V_j$.  Let us
denote $e_\ell=\{v_{i,x},v_{j,y}\}$, for $1\le \ell\le m$ those edges,
assuming there are $m$ of them, with $v_{i,x}\in V_i$ and $v_{j,y}\in
V_j$.  For each such edge $e_\ell$, there are 6
independent paths in $H$, denoted $a_\ell, a'_\ell, b_\ell, b'_\ell,
e_\ell, e'_\ell$, that cover a circular track $T_{i,j}$ (see
Figure~\ref{fig:W1-edge}). The parts leaving the track will be defined
later, but these parts will not induce new edges among those paths. As
these paths can be partitioned into 6 cliques (one for the paths
$a_\ell$, one for the paths $a'_\ell$, $\ldots$), an independent set
uses at most 6 of these paths. Furthermore, one can see that if it
uses 6 of them, then those are the paths $a_\ell, a'_\ell, b_\ell,
b'_\ell, e_\ell, e'_\ell$ for a given $1\le \ell \le m$. In such a
case we say that the edge $e_\ell$ of $G$ was {\bf selected}.

\begin{figure}[h!]
\begin{center}
\includegraphics[width=0.9\textwidth]{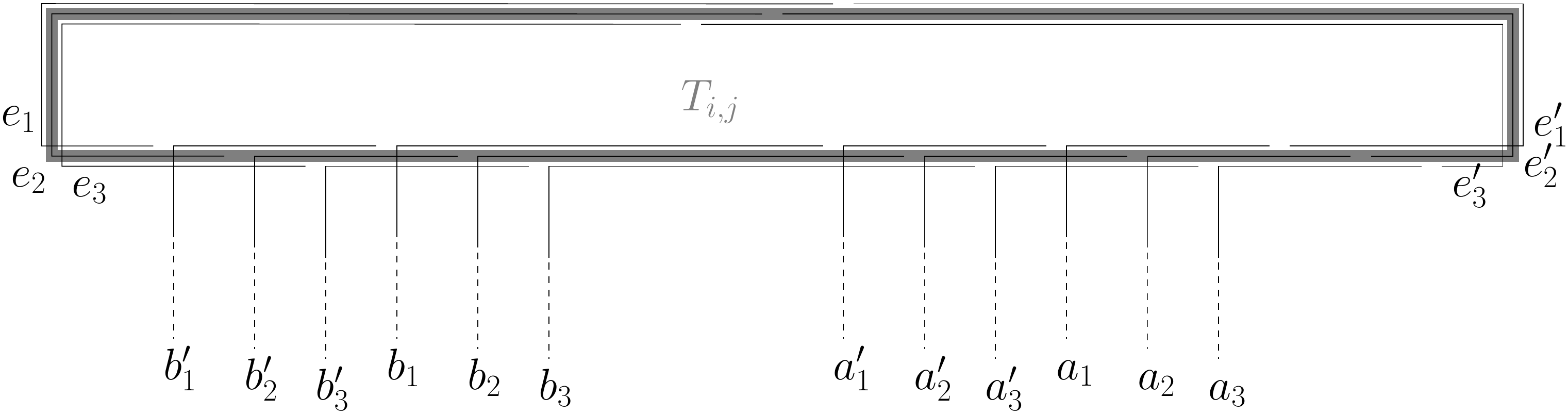}
\end{center}
\caption{An edge gadget on the circular track $T_{i,j}$, \ie~for the
  edges between $V_i$ and $V_j$. Here there are only 3 such edges.}
\label{fig:W1-edge}
\end{figure}

Let us now define a {\bf vertex gadget}. For any integer $i$ such
that $ 1\le i\le k$, we construct a gadget in $H$ that corresponds to
the vertices of the set $V_i$.  Let us denote $v_{i,x}$ these vertices,
for $1\le x\le |V_i|$.  For each such vertex $v_{i,x}$, there are 3
independent paths in $H$, denoted $v_{i,x}, v'_{i,x}, v''_{i,x}$, that
partially cover a circular track $T_{i}$ (see
Figure~\ref{fig:W1-vertex}). As these paths can be partitioned into 3
cliques (one for the paths $v_{i,x}$, one for the paths $v'_{i,x}$, and
one for the paths $ v''_{i,x}$), an independent set uses at most 3 of
these paths.

\begin{figure}[h!]
\begin{center}
\includegraphics[width=0.9\textwidth]{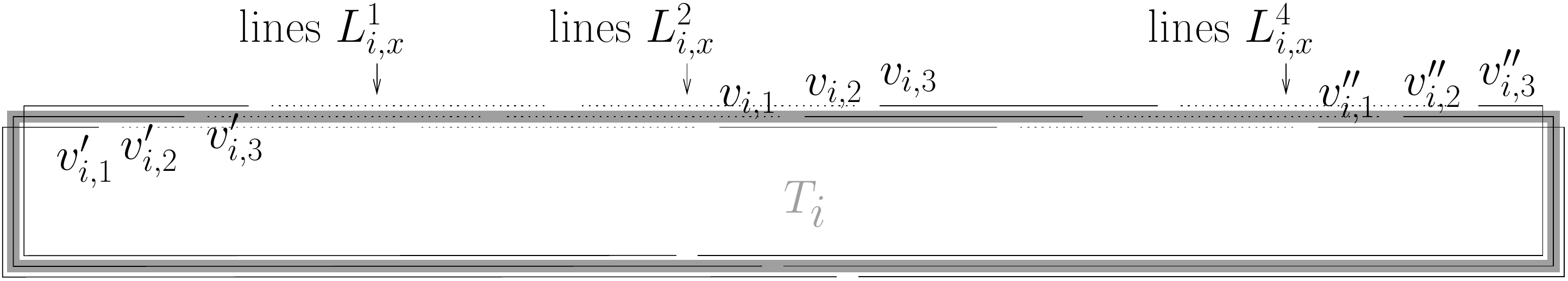}
\end{center}
\caption{A vertex gadget on the circular track $T_{i}$ (\ie~for the
  vertices in $V_i$), and the dotted lines $v^j_{i,x}$. Here $k=4$ and $i=3$.}
\label{fig:W1-vertex}
\end{figure}

In Figure~\ref{fig:W1-vertex}, there are also dotted lines $v^j_{i,x}$
for $1\le x\le |V_i|$ and $1\le j \le k$ with $j\neq i$. These lines
will serve as landmarks to describe the end of the paths $a_\ell,
a'_\ell$ (resp. $b_\ell,b'_\ell$) coming from the edge gadgets
$T_{i,j}$ with $i< j\le k$ (resp. from the edge gadgets $T_{j,i}$ with
$1\le j<i$). See Figure~\ref{fig:W1-connection} for a description of
how these paths connect to a vertex gadget. The idea is that the paths
$a_\ell$, and $a'_\ell$ coming from $T_{i,j}$ go to $T_{i}$ and cover
the ends of the line $v^j_{i,x}$ if and only if the edge $e_\ell$ is
incident to vertex $v_{i,x}\in V_i$. The connection between the edge
gadget on $T_{i,j}$ and the vertex gadget on $T_{j}$ is similar but
uses paths $b_\ell$, and $b'_\ell$ instead.

\begin{figure}[h!]
\begin{center}
\includegraphics[width=0.9\textwidth]{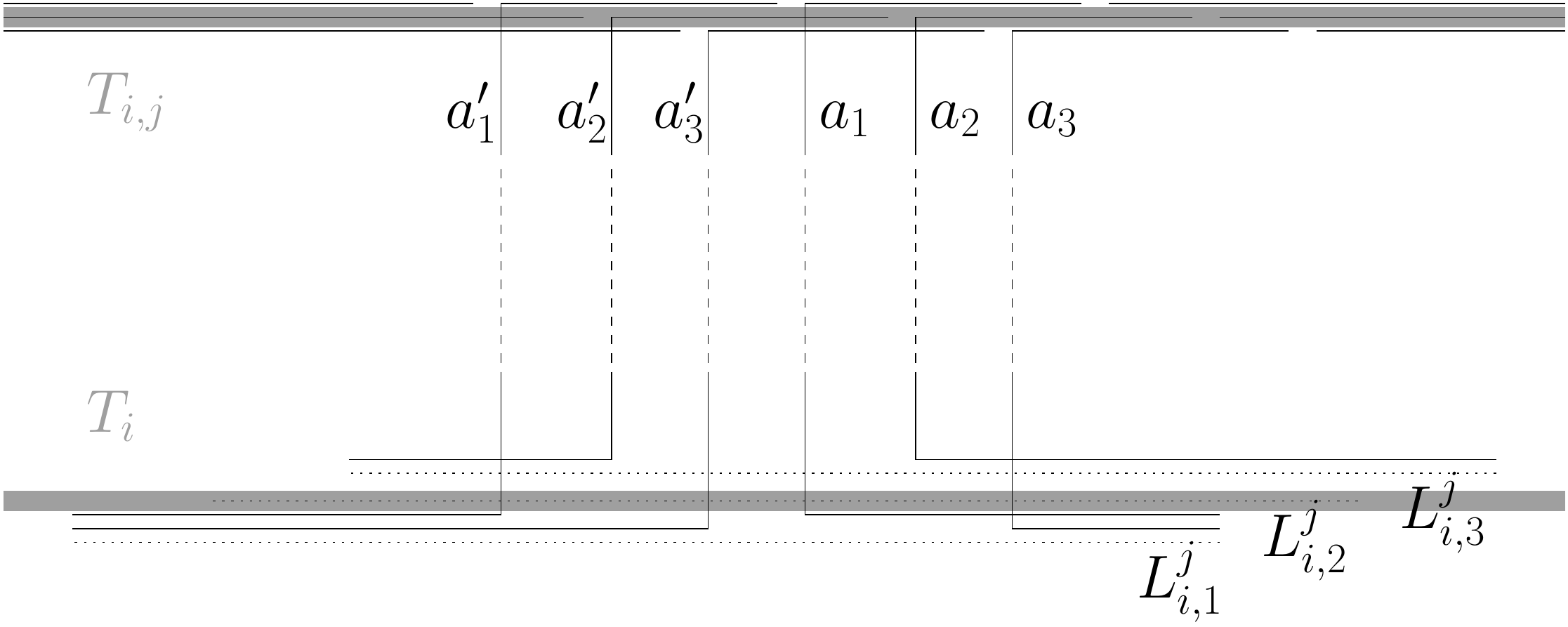}
\end{center}
\caption{Connection between paths $a_\ell$ and $a'_\ell$ of an edge
  gadget on $T_{i,j}$ to the vertex gadget on $T_{i}$. Here we assume
  that $e_1$ and $e_3$ are incident to $v_{i,1}$, while $e_2$ is incident
  to $v_{i,3}$, where $V_i = \{v_{i,1},v_{i,2},v_{i,3}\}$.}
\label{fig:W1-connection}
\end{figure}

One can see in Figure~\ref{fig:W1-overall}, how to arrange all these
gadgets in order to make the above mentioned connections, and no
other. Note that the paths $a_\ell, a'_\ell, b_\ell$ and $b'_\ell$ leaving $T_{1,2}$ cross several gadgets, without creating any edge, before entering gadget $T_{1}$ or $T_{2}$. 

\begin{figure}[htbp!]
\begin{center}
\includegraphics[width=1\textwidth]{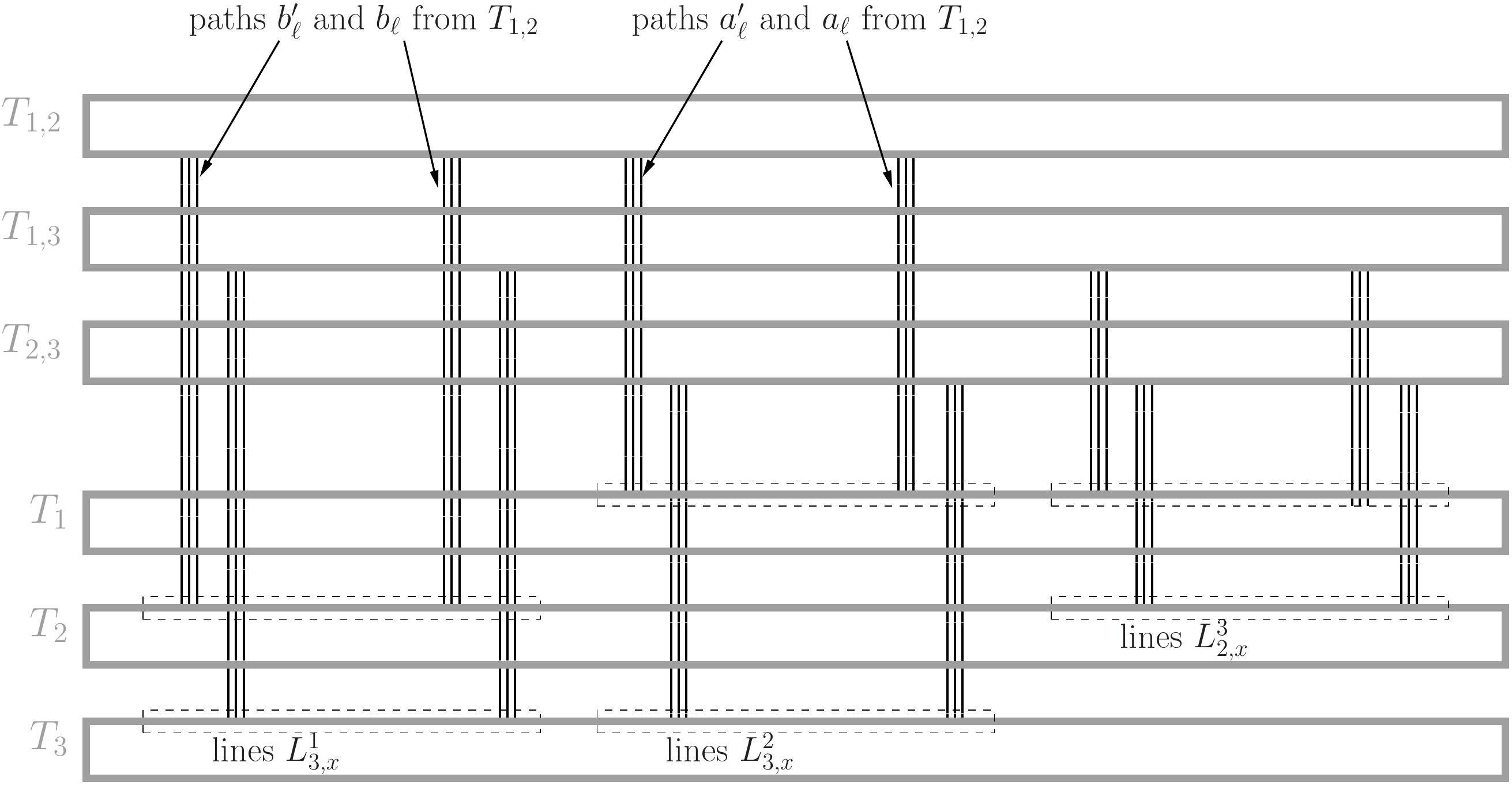}
\end{center}
\caption{Overall arrangement of the gadgets. Here $k=3$.}
\label{fig:W1-overall}
\end{figure}

Hence, one can see that an independent set of size $3k^2$ (hence using
one path from each of the above mentioned cliques, that is using 6
vertices from each edge gadget and 3 vertices from each vertex gadget)
selects one edge from each pair $i,j$ with $1\le i<j\le k$. Suppose
that between $V_i$ and $V_j$ the edge $e_\ell$ is selected. This
implies that at the vertex gadget $T_i$ (resp. $T_j$), the paths
$a_\ell, a'_\ell$ (resp. $b_\ell, b'_\ell$) of $T_{i,j}$, that are in
the independent set, cover the line $v^j_{i,x}$ (resp. $v^i_{j,y}$)
for some $1\le x\le |V_i|$ (resp. $1\le y\le |V_j|$) such that $e_\ell
= (v_{i,x},v_{j,y})$.  Then on $T_i$ (resp. $T_j$) the paths in the
independent set are necessarily $v_{i,x}$, $v'_{i,x}$ and $v''_{i,x}$
(resp. $v_{j,y}$, $v'_{j,y}$ and $v''_{j,y}$). Hence the vertex
$v_{i,x}\in V_i$ (resp. $v_{j,y}\in V_j$) is selected. This implies
that $H$ has an independent set of size $3k^2$ if and only if $G$ has
a $k$-clique (induced by the selected edges and vertices), and this
concludes the proof of the theorem.  \qed\end{proof}

\newpage
\bibliographystyle{plain}
\bibliography{biblio.bib}

\end{document}